\documentclass[conference,letterpaper]{IEEEtran}
\usepackage[cmex10]{amsmath} % Use the [cmex10] option to ensure complicance
\usepackage{algorithm}
\usepackage{algorithmicx}
\usepackage{algpseudocode}
\usepackage{amssymb}
\usepackage{amsfonts}
\usepackage{arydshln} 
\usepackage{bm}
\usepackage{color}
\usepackage{graphicx}
\usepackage{multirow} 
\usepackage{soul}
\usepackage{amssymb}
\usepackage{subfigure}
\usepackage{float}
\usepackage{mathtools}
\usepackage[utf8]{inputenc}
\usepackage[english]{babel}
\usepackage{amsthm}
\usepackage[nopostdot,nogroupskip,nonumberlist]{glossaries}
\usepackage{cite}
\usepackage{lipsum}
\usepackage{clipboard}
\usepackage{booktabs}
\usepackage{float}
\usepackage[utf8]{inputenc} 
\usepackage[T1]{fontenc}
\usepackage{url}
\usepackage{ifthen}
\newacronym{CACC}{CACC}{cooperative adaptive cruise control}
\newacronym{dCACC}{dCACC}{degraded cooperative adaptive cruise control}
\newacronym{ACC}{ACC}{adaptive cruise control}
\newacronym{CGKF}{CGKF}{constant gain Kalman filter}
\newacronym{ZOH}{ZOH}{zero-order-hold}
\newacronym{MSE}{MSE}{mean square error}
\newacronym{MMSE}{MMSE}{minimum mean square error}
\newacronym{MIMO}{MIMO}{multiple-input-multiple-output}
\newacronym{ISAC}{ISAC}{integrated sensing and communication}
\newacronym{BS}{BS}{base station}
\newacronym{iid}{i.i.d}{independent and identically distributed}
\newacronym{DD}{DD}{data-dependent}
\newacronym{DI}{DI}{data-independent}
\newacronym{wrt}{w.r.t.}{with respect to}
\newacronym{SVD}{SVD}{singular value decomposition}
\newacronym{DFT}{DFT}{discrete Fourier transform}
\newacronym{SAA}{SAA}{sample average approximation}
\newacronym{UE}{UE}{user equipment}
\newacronym{ADC}{ADC}{anolog-to-digital converter}
\newacronym{TIR}{TIR}{target impulse response}

\usepackage[switch]{lineno}
\usepackage[bookmarks,colorlinks,bookmarksnumbered]{hyperref}
\hypersetup{colorlinks=true,pdfauthor=author}
\DeclarePairedDelimiter\abs{\lvert}{\rvert}%
\DeclarePairedDelimiter\norm{\lVert}{\rVert}%

\newcommand\myvec[1]{\bm{#1}}
\newcommand\mymat[1]{\mathbf{#1}}
\newcommand\mymatb[1]{\bar{\mathbf{#1}}}

\newcommand\myvech[1]{\hat{\bm{#1}}}

\newcommand\myexpect[1]{\mathbb{E}\left\{#1\right\}}

\newcommand{\myvtr}{\mbox{vec}}
\newcommand{\mytr}{\mbox{Tr}}
\newcommand\myexpecttheta[1]{\mathbb{E}_{\mymat{\Theta}}\left\{#1\right\}}
\newcommand\myexpectd[1]{{E}_{D}\left(#1\right)}
\newcommand\mytrop[1]{\operatorname{Tr}\left(#1\right)}
\makeatletter
\let\oldabs\abs
\def\abs{\@ifstar{\oldabs}{\oldabs*}}
\let\oldnorm\norm
\def\norm{\@ifstar{\oldnorm}{\oldnorm*}}
\makeatother

\makeatletter
\renewcommand{\fnum@figure}{Fig. \thefigure}
\makeatother

\newtheorem{theorem}{Theorem}
\newtheorem{lemma}{Lemma}

\newtheorem{remark}{\it{Remark}}

\let\oldtheorem\theorem
\let\endoldtheorem\endtheorem

\let\oldremark\remark
\let\endoldremark\endremark
\renewenvironment{theorem}
{\oldtheorem\normalfont}
{\endoldtheorem}
\renewenvironment{remark}
{\oldremark\normalfont}
{\endoldremark}

\begin{document}
\title{Task-Based Quantizer Design for Sensing With Random Signals}

\author{%
	\IEEEauthorblockN{Hang Ruan, Fan Liu}
	\IEEEauthorblockA{School of System Design and Intelligent Manufacturing\\ Southern University of Science and Technology\\ Shenzhen 518055, China\\
		Email: ruanhang\_ee@outlook.com, liuf6@sustech.edu.cn}
}

%\author{
%	Hang~Ruan, %~\IEEEmembership{Student~Member,~IEEE,}
%	%Hongbin~Li,~\IEEEmembership{Senior Member,~IEEE,}
%	Tianyao~Huang,
%	Yimin~Liu,~\IEEEmembership{Member,~IEEE,}	
%	and~Xiqin~Wang
%	% <-this % stops a space
%	\thanks{
%		H. Ruan, T. Huang, Y. Liu and X. Wang are with the Department of Electronic
%		Engineering, Tsinghua University, Beijing, 100084 China (email:
%		ruanh15@mails.tsinghua.edu.cn; huangtianyao@tsinghua.edu.cn; yiminliu@tsinghua.edu.cn;  wangxq\_ee@tsinghua.edu.cn).
%		T. Huang is the corresponding author.
%	}
%	\thanks{
%		This work was supported by the National Natural Science Foundation of China (Grant No. 61801258).
%	}
%}

\maketitle

\begin{abstract}
	In integrated sensing and communication (ISAC) systems, random signaling is used to convey useful information as well as sense the environment. Such randomness poses challenges in various components in sensing signal processing. In this paper, we investigate quantizer design for sensing in ISAC systems. Unlike quantizers for channel estimation in massive multiple-in-multiple-out (MIMO) communication systems, sensing in ISAC systems need to deal with random nonorthogonal transmitted signals rather than a fixed orthogonal pilot. Considering sensing performance and hardware implementation, we focus on task-based hardware-limited quantization with spatial analog combining. We propose two strategies of quantizer optimization, i.e., data-dependent (DD) and data-independent (DI). The former achieves optimized sensing performance with high implementation overhead. To reduce hardware complexity, the latter optimizes the quantizer with respect to the random signal from a stochastic perspective. We derive the optimal quantizers for both strategies and formulate an algorithm based on sample average approximation (SAA)  to solve the optimization in the DI strategy. Numerical results show that the optimized quantizers outperform digital-only quantizers in terms of sensing performance. Additionally, the DI strategy, despite its lower computational complexity compared to the DD strategy, achieves near-optimal sensing performance.
%	\newline
%	\textbf{Keywords:} Integrated Sensing and Communication (ISAC), MIMO systems, quantizer design, random signaling, sensing performance optimization, Mean Squared Error (MSE), Target Impulse Response (TIR), Data-Dependent (DD), Data-Independent (DI), Singular Value Decomposition (SVD), Stochastic Average Approximation (SAA).
\end{abstract}

\section{Introduction}
\label{sec:intro}
In the next-generation wireless networks, sensing assumes a vital role for applications including intelligent transportation, smart manufacturing, smart cities and public safety. \Gls{ISAC} is identified as a key technology to achieve ubiquitous sensing with communication systems \cite{liu2022integrated}. \Gls{ISAC} aims at implementing the functionalities of sensing and communication on the shared wireless resources and hardware, which is enabled by their similarities in hardware architectures, channel characteristics and signal processing pipelines, thereby offering benefits including reduced hardware cost, improved spectral and energy efficiency, and reduced latency \cite{liu2022integrated}. 
More recently, \gls{MIMO} \gls{ISAC} has gained growing research
interests, driven by the benefits including the spatial multiplexing and diversity gains in communication, and the waveform and spatial diversity gains in sensing \cite{gao2022integrated}.
In a typical \gls{ISAC} paradigm \cite{liu2022integrated}, an \gls{ISAC} transmitter, e.g., a \gls{BS},  emits a waveform consisting of a sequence of snapshots to communicate with other communication devices. %, e.g., \glspl{UE}.
This waveform is also received by a sensing receiver, e.g., the same or another \gls{BS}, after propagating in a sensing channel, which is exploited to extract information of the environment.

A unique characteristic of \gls{ISAC} is that the transmitted signals must be random to convey useful information \cite{lu2023random, zhang2021overview}. Such randomness can be realized by random selection from certain codebooks \cite{zhang2021overview}. In contrast, conventional sensing systems, e.g., radars, usually use deterministic signals with specific properties, e.g., narrow mainlobe, high peak-to-sidelobe level ratio (PSLR) and orthogonality \cite{xie2023novel}. These properties may be compromised due to random signaling in \gls{ISAC}. Thus, the sensing performance may be degraded when directly applying conventional sensing techniques based on deterministic signals to \gls{ISAC} sensing. For example, \cite{lu2023random} investigates the precoding design in \gls{ISAC} systems, revealing that conventional precoding scheme for deterministic signals leads to higher \gls{MSE} for sensing with random \gls{ISAC} signals. %Thus, the authors propose new precoding schemes for \gls{ISAC} random signals. 
In the same spirit, we identify that the quantizer in \gls{ISAC} systems may also need dedicated design due to random signaling, %design in sensing receiver as another component that is affected by \gls{ISAC} random signaling, 
which motivates this work. %the study of this paper.

In \gls{ISAC}, the sensing task is mainly accomplished in the digital domain, which requires the received analog signal to be quantized into digital presentations before further signal processing \cite{gray1998quantization}. More relevant to this work is the quantizer design for channel estimation in massive \gls{MIMO} networks \cite{shlezinger2019asymptotic, ma2021bit}, where %the similarity lies in that both tasks aim at estimating a \gls{MIMO} channel. In \cite{shlezinger2019asymptotic, ma2021bit}, 
the following three aspects are elaborated: \textit{1) Task-ignorant quantization and task-based quantization}: Typically, the signal is quantized following the criterion that some general distortion measure, e.g., \gls{MSE}, between the analog and digital signal is minimized, whereas ignoring the specific aim of the system \cite{gray1998quantization}, referred to as task-ignorant quantization; In massive \gls{MIMO} networks, however, the system's task is to estimate the channel rather than to recover the analog signal itself \cite{shlezinger2019asymptotic, ma2021bit}. By designing the quantizers oriented to the specific task, referred to as task-based quantization, better sensing performance may be achieved \cite{shlezinger2019asymptotic}. \textit{2) Vector quantization and scalar quantization}: Vector quantization is proven to achieve better performance of the task \cite{shlezinger2019hardware}, but it is infeasible for real-time applications with high-dimensional inputs due to its computational and hardware complexity. In this case, it is more practical to implement scalar quantization. Specifically, \cite{shlezinger2019asymptotic} and \cite{ma2021bit} exploit the scalar quantization with uniform \glspl{ADC}, also referred to as hardware-limited task-based quantization. \textit{3) Temporal analog combining and spatial analog combining}: Due to the storage limitation of the sensing receiver, it is usually impractical to collect all the raw data and then quantize the whole signal, referred to as temporal analog combining; Instead, it is more desirable to quantize the signal from all antennas snapshot by snapshot, also referred to as spatial analog combining, which has been investigated in \cite{shlezinger2019asymptotic}. These three aspects are also applicable to quantization in \gls{ISAC} systems, motivating us to consider hardware-limited task-based quantization with spatial analog combining only.

%Task-based quantization can be categorized into two methods, namely, vector quantization and scalar quantization \cite{shlezinger2019hardware}. Vector quantization is proven to achieve better performance of the task \cite{shlezinger2019hardware}, but it is infeasible for real-time applications with high-dimensional inputs due to its computational and hardware complexity. In this case, it is more practical to implement scalar quantization. Specifically, in this paper, we consider the scalar quantization with uniform \glspl{ADC}, which have been widely used in both industry and academia \cite{gray1998quantization,shlezinger2019hardware} and also referred to as \textit{hardware-limited task-based quantization}.  

%Additionally, due to the storage limitation of the sensing receiver, it is usually impractical to collect the echoes of all snapshots before quantization; Instead, it is more desirable to quantize the signal from all antennas snapshot by snapshot, also referred to as \textit{spatial analog combining}. To this end, \cite{shlezinger2019asymptotic} and \cite{ma2021bit} proposed task-based quantization with only spatial analog combining for the task of channel estimation in massive \gls{MIMO} networks. This task is similar to sensing in \gls{ISAC} systems since both tasks aim at estimating a \gls{MIMO} channel. 

Nevertheless, the quantization methods in \cite{shlezinger2019asymptotic, ma2021bit} cannot be directly applied to \gls{ISAC} sensing since they generally assume fixed pilot signals with orthogonality. % , based on which the quantizers are designed. 
In the case of \gls{ISAC} systems, however, the transmitted signals are random %must be random to convey useful information \cite{lu2023random, zhang2021overview}. Such randomness can be realized by selection from certain codebooks, frequency hopping, etc \cite{zhang2021overview}. Moreover, 
and the orthogonality cannot be guaranteed, as aforementioned.
% in the random signal, which also makes the methods  in \cite{shlezinger2019asymptotic, ma2021bit} inapplicable to \gls{ISAC} sensing.

To address the challenge in quantization of \gls{ISAC} systems, this paper investigates the task-based hardware-limited quantization in the receiver with random transmitted signals. %The sensing functionality is performed with \gls{MIMO}. 
The task of sensing receiver is to estimate the \gls{TIR}, whose quantizer is designed to minimize the \gls{MSE} of \gls{TIR} estimation. Our main contributions  are as follows:

1) We formulate the system and signal model of sensing receiver in \gls{MIMO} \gls{ISAC} systems. Compared to \cite{shlezinger2019asymptotic, ma2021bit}, the signal randomness and the correlation at the transmitting antennas are considered in the modeling.
	
 2) We propose two strategies for quantizer design. The quantizer consists of a pre-processing matrix, a sequence of scalar \glspl{ADC} and a post-processing matrix. The two matrices are optimized with the criterion of minimizing the \gls{MSE} of \gls{TIR} estimation. Facing the randomness of \gls{ISAC} signals, we propose both \gls{DD} and \gls{DI} strategies for quantizer optimization. In the \gls{DD} strategy, the pre-processing matrix is designed differently for each realization of transmitted signal, where the \gls{MSE} is averaged over the  \gls{TIR} and receiver noise, resulting in high implementation overhead. In the \gls{DI} strategy, a fixed pre-processing matrix is optimized by minimizing the \gls{MSE} \gls{wrt} \gls{TIR}, receiver noise and the transmitted signal, which may be readily implemented in practice.
 
3) We derive the solution of quantizer optimization. We prove that the optimal quantizer can be obtained by a combination of \gls{SVD}, eigen-decomposition and convex optimization. Then, we propose an algorithm based on \gls{SAA} to tackle the stochastic optimization problem in the \gls{DI} strategy.

4) We evaluate the proposed quantizer design strategies with numerical results. Our results demonstrate that the DI strategy can achieve sensing performance close to DD with lower implementation overhead. It is also demonstrated that the proposed quantizers outperform digital-only quantization in terms of sensing performance. 

The remainder of this paper is arranged as follows: In Section \ref{sec:model}, we formulate the system and signal model of the \gls{ISAC} system. In Section \ref{sec:quantizer_structure}, we construct the structure of quantizer in the sensing receiver and optimize the quantizer. In Section \ref{sec:result}, we present the numerical results. In Section \ref{sec:discussion}, we present the conclusion and discussion.

\textbf{Notations}: Boldface lower-case letters, e.g., $\bm{x}$, denote vectors; The $i$th element of $\bm{x}$ is written as $(\bm{x})_i$. Boldface upper-case letters, e.g., $\mathbf{M}$, denote matrices and $(\mathbf{M})_{i,j}$ is its $(i, j)$th element. 
%Sets are denoted with calligraphic letters, e.g., $\mathcal{X}$, and $\mathcal{X}^n$ is the $n$th order Cartesian power of $\mathcal{X}$.
Transpose, Hermitian transpose, complex conjugate, vectorization, Euclidean norm, trace, stochastic expectation, real part, imaginary part, sign and rounding toward negative infinity %and mutual information 
are written as $(\cdot)^T$, $(\cdot)^H$, $(\cdot)^*$, $\myvtr(\cdot)$, $\|\cdot\|$, $\text{Tr}(\cdot)$, $\mathbb{E}\{\cdot\}$, $\operatorname{Re}(\cdot)$, $\operatorname{Im}(\cdot)$, $\text{sign}(\cdot)$ and $\lfloor\cdot\rfloor$, 
 %and $I(\cdot;\cdot)$, 
 respectively, and $\mathbb{R}$ and $\mathbb{C}$ are the domains of real and complex numbers, respectively. The Kronecker product operator between two matrices is denoted by $\otimes$. For a semi-positive definite matrix $\mymat{A} \in \mathbb{C}^{n\times n}$, $\mymat{A}^{\frac{1}{2}} \in \mathbb{C}^{n\times n}$ denotes the matrix such that $\mymat{A}^{\frac{1}{2}} (\mymat A^{\frac{1}{2}})^H = (\mymat{A}^{\frac{1}{2}})^H \mymat{A}^{\frac{1}{2}} = \mymat{A}$. We use %$a^+$ to denote $\max(a, 0)$, and 
 $\mathbf{I}_n$ to denote the $n \times n$ identity matrix. 
 %All logarithms are taken to basis 2.

\section{System and signal model}
\label{sec:model}
%\subsection{Signal model}
%\label{subsec:model}
We consider a \gls{MIMO} \gls{ISAC} system. The \gls{BS} is equipped with a pair of transmitting and receiving antenna arrays. The number of transmitting and receiving antennas are $N_t$ and $N_r$, respectively. In each frame, the \gls{BS} transmits a signal consisting of $L$ snapshots to communicate with other devices (We assume $L\geq N_t$). The transmitted signal is denoted by $\mymat{\Theta} = [\myvec{\theta}_1, \dots, \myvec{\theta}_L]\in\mathbb{C}^{N_t\times L}$. The signal of each snapshot, $\myvec{\theta}_l \in \mathbb{C}^{N_t}$, are \gls{iid} variables following $\mathcal{CN}(\myvec{0}, \mymat{R}_{\theta})$. %  as complex Gaussian vectors with zero mean and covariance $\mymat{R}_{\theta}$. 
Additionally, $\mymat{\Theta}$ is reflected by the environment and the echo $\mymat{Y} = [\myvec{y}_1, \dots, \myvec{y}_L] \in\mathbb{C}^{N_t\times L}$ is received by the receiving array, which is used to sense the environment. The received signal $\mymat{Y}$ is given by
%In some applications, the receiver cannot storage the received signal of all snapshots due to the limited storage room. Therefore, I am considering quantizing the received signal \textit{snapshot by snapshot}. 
\vspace{-0.2cm}
\begin{equation}
\label{eqn:system_model}
\mymat{Y} = \mymat{G}\mymat{\Theta}+\mymat{W},
\vspace{-0.2cm}
\end{equation}
where %$\mymat{X}\in\mathbb{C}^{N_t\times L}$ is the transmitted signal; 
$\mymat{G}\in\mathbb{C}^{N_r\times N_t}$ is the \gls{TIR} to be estimated, which is supposed to follow the Kronecker model, written as \cite{weichselberger2006stochastic}
\vspace{-0.2cm}
\begin{equation}
	\mymat{G} = \mymat{R}_A^{\frac{1}{2}} \mymat{G}_0 \left(\mymat{R}_B^{\frac{1}{2}}\right)^T,
\vspace{-0.2cm}
\end{equation}
where $\mymat{R}_A \in \mathbb{C}^{N_r \times N_r}$ and $\mymat{R}_B \in \mathbb{C}^{N_t \times N_t}$ are the spatial correlation matrices at the receiving and transmitting array, respectively. The entries of matrix $\mymat{G}_0 \in \mathbb{C}^{N_r \times N_t}$ are \gls{iid} variables following $\mathcal{CN}(0,1)$. By letting $\myvec{g} = \myvtr (\mymat{G}) \in \mathbb{C}^{N_t N_r}$, the correlation of channel may be presented as follows \cite{kermoal2002stochastic}:
\vspace{-0.2cm}
\begin{equation}
\mymat{R}_g = \mathbb{E}\{\myvec{g}\myvec{g}^H\} = \mymat{R}_B \otimes \mymat{R}_A.
\vspace{-0.2cm}
\end{equation}
 %$\mymat{Y}\in\mathbb{C}^{N_t\times L}$ is the received signal;
The matrix $\mymat{W} \in \mathbb{C}^{N_t \times L}$ denotes the receiver noise. Considering the spatial correlation of the receiving array, its $l-$th column $\myvec{w}_l \in \mathbb{C}^{N_t}$ follows $\mathcal{CN}(\myvec{0}, \sigma_w^2 \mymat{R}_A)$ in an \gls{iid} manner. %is \gls{iid} as zero-mean complex Gaussian vectors with correlation $\sigma_w^2 \mymat{R}_A$. %The number of transmitting and receiving antennas are denoted by $N_t$ and $N_r$, respectively; $L$ is the number of snapshots. The task is to estimate $\mymat{H}$ from $\mymat{Y}$

Letting $\myvec{y} = \myvtr (\mymat{Y})$ and $\myvec{w} = \myvtr (\mymat{W})$ transforms \eqref{eqn:system_model} as
\vspace{-0.2cm}
\begin{equation}
	\myvec{y} = \left(\mymat{\Theta}^T\otimes \mymat{I}_{N_r}\right)\myvec{g}+\myvec{w}.
\vspace{-0.2cm}
\end{equation}

\begin{remark}
	Our formulation of signal model is more general than that in \cite{shlezinger2019asymptotic} from the following three aspects: First, the transmitted signals are random \gls{ISAC} signals rather than fixed pilots in \cite{shlezinger2019asymptotic}; Second, the transmitted signal is not necessarily orthogonal, i.e., $\mymat{\Theta \Theta}^H \neq L\mymat{I}_{N_t}$; Third, the spatial correlation of the transmission is also involved, i.e., we do not require $\mymat{R_B}$ to be diagonal. %, which is assumed in \cite{shlezinger2019asymptotic}.
	\vspace{-0.2cm}
\end{remark}

The received signal $\myvec{y}$ is quantized into a vector $\myvec{z}$ of length $P$, which is encoded with a quantization rate $R$ \cite{shlezinger2019hardware}, corresponding to  $M_{\text{bit}} = R N_r L$ bits at maximal, which depends on the available memory. Then, the \gls{BS} attempts to attain an estimate of \gls{TIR} $\myvec{g}$ by leveraging $\myvec{z}$, denoted by $\myvech{g}$.

\section{Quantizer structure and optimization}
\label{sec:quantizer_structure}
In this section, we construct the structure of quantizer in the sensing receiver and derive its optimum. To this end, we first formulate the quantizer structure based on some hardware considerations; Second, we formulate the \gls{MSE} of \gls{TIR} estimation as the performance metric of quantizer design; Third, we formulate the quantizer optimization in the strategies of \gls{DD} and \gls{DI}; Finally, we provide a theorem on the optimal quantizer and an \gls{SAA}-based algorithm to solve the stochastic optimization involved in the \gls{DI} strategy.

As discussed in Section \ref{sec:intro}, when designing the quantizer, several aspects of hardware limitation \gls{BS} need to be considered: First, vector quantization may not be feasible in practice, especially for massive \gls{MIMO}; Instead, scalar quantizers are easier to implement in applications. Second, 	the \gls{BS} may not be able to storage all the snapshots in the analog domain due to the memory limitation; Thus, it is desirable to quantize the signal with spatial analog combining snapshot by snapshot.
Considering these limitations, we consider the receiver with similar structure to \cite{shlezinger2019asymptotic}, whose three steps are as follows:

	\textbf{1) Spatial analog combining}: For the received signal in each snapshot, i.e., $\myvec{y}_l$, spatial analog combining is performed with a pre-processing matrix $\mymat{A} \in \mathbb{C}^{\tilde{P}\times N_r}$, yielding $\myvec{u}_l = \mymat{A}\myvec{y}_l$, where $\tilde{P} = P/L$ is the dimension of combining output. According to \cite[Corollary 1]{shlezinger2019hardware}, $\tilde{P}\leq N_r$ is assumed in order to minimize the \gls{MSE} of \gls{TIR} estimation. After the stacking %$\myvec{u}_l$ from all snapshots into $\myvec{u}\in \mathbb{C}^{L\tilde{P}}$, i.e.,
	 $\myvec{u} = \myvtr([\myvec{u}_1,\dots, \myvec{u}_L])\in \mathbb{C}^{L\tilde{P}}$, this step can be rewritten as 
	 \vspace{-0.2cm}
	\begin{equation}
	\myvec{u} = (\mymat{I}_L \otimes \mymat{A}) \myvec{y}.
	\vspace{-0.2cm}
	\end{equation}
	
	\textbf{2) Scalar quantization}: The $L\tilde{P}$ entries of $\myvec{u}$ are fed into $L\tilde{P}$ identical scalar dithered quantizers \cite{gray1993dithered} with resolution $\tilde{M}$, yielding the quantized vector $\myvec{z} \in \mathbb{C}^{L\tilde{P}}$, i.e.,
	\vspace{-0.2cm}
	\begin{equation}
	(\myvec{z})_i = Q_{\tilde{M}, K_d} \left((\myvec{u}\right)_i),
	\vspace{-0.2cm}
	\end{equation}
	where $Q_{\tilde{M}, K_d} (\cdot)$ denotes the scalar quantizer with resolution $\tilde{M}$ and $K_d$ dither signals. Given $M_{\text{bit}}$ bits  and $\tilde{P}$ scalar \glspl{ADC}, $\tilde{M}$ is given by $\tilde{M} = \left\lfloor 2^{\frac{M_{\text{bit}}}{\tilde{P}L}} \right\rfloor$. We consider the dithered quantizer since when $K_d \geq 2$ and the input plus dither signals is within the quantizer's support $\gamma$, such choice enables the quantizer's output to be written as the sum of the input and an uncorrelated white quantization noise \cite{gray1993dithered}, thereby greatly facilitating the analysis on system performance. Additionally, this property of dithered quantizers is also approximately satisfied in uniform quantizers without dithering with a wide range of input distributions including Gaussian \cite{widrow1996statistical}. %, which will be further examined with numerical results in Section \ref{sec:result}. 
	The quantization spacing is $\Delta = 2\gamma/\tilde{M}$. When the dithered quantizer's input is $x\in \mathbb{C}$, the output is  
	\vspace{-0.2cm}
	\begin{equation} \label{eqn:dither_scalar}
	\begin{aligned}
	Q_{\tilde{M}} (x) \! = \! q\!\left(\operatorname{Re}\!\left\{x\!+\!\sum_{k=1}^{K_d}\!\xi_k\right\}\!\right)
	\!+\!j  q\!\left(\operatorname{Im}\!\left\{x\!+\!\sum_{k=1}^{K_d}\!\xi_k\!\right\}\!\right), 
	\end{aligned}
	\vspace{-0.2cm}
	\end{equation}
	where $\{\xi_k\}_{k=1}^{\tilde{K_d}}$ are complex random variables independent of input $x$, with \gls{iid} real and imaginary parts uniformly distributed over $\left[-{\Delta}/{2}, {\Delta}/{2}\right]$, and $q(\cdot)$ is a uniform quantizer given by
	\begin{equation}
	q(\alpha)= \begin{cases}-\gamma+\Delta\left(l+\frac 12\right), & \alpha-l \cdot \Delta+\gamma \in[0, \Delta], \\ \operatorname{sign}(\alpha)\left(\gamma-\frac{\Delta}{2}\right), & |\alpha|>\gamma. \end{cases}
	\end{equation}
	The setting of support $\gamma$ should guarantee that the dithered input is within the operating range $[-\gamma, \gamma]$ with sufficiently high probability. To this end, $\gamma$ shall be defined as some multiple $\eta$ of the maximal standard deviation of the input $\myvec{u}$ plus dither signals, denoted by $\myvec{\xi}_1,\dots, \myvec{\xi}_{K_d}$, i.e., 
	\begin{equation}
	\label{eqn:support_multiple}
	\gamma = \eta \sqrt{\max_{i = 1,\dots, L\tilde{P}}\myexpect{\left(\myvec{u}+\sum_{k=1}^{K_d} \myvec{\xi}_k\right)_{i.i}^2}}.
	\end{equation}
	In the following analysis, we assume that dithered input is within the operating range with probability $1$.
	
	\textbf{3) Digital processing}: The estimate $\myvech{g}$ is obtained by feeding $\myvec{z}$ into a post-processing matrix $\mymat{B}\in \mathbb{C}^{N_tN_r \times L\tilde{P}}$, i.e.
	\begin{equation}
	\myvech{g} = \mymat{B}\myvec{z}.
	\end{equation}

%\subsection{Quantizer optimization}
%\label{subsec:quantizer_opt}
The purpose of quantizer design is to find the optimal $\mymat{A}$ and $\mymat{B}$ such that the distortion between $\myvec{g}$ and $\myvech{g}$, quantified with average \gls{MSE} in this paper, reaches its minimum, i.e.,
\begin{equation} 
\label{eqn:mse_g}
\min _{\mymat{A}, \mymat{B}} \sigma^2_g \triangleq \frac{1}{N_t N_r}\mathbb{E}\left\{\norm{\myvech{g}-\myvec{g}}^2\right\}.
\end{equation}

\begin{remark}
	\label{remark:tradeoff}
	When the quantization rate $R$ is fixed, there is a trade-off between the number of scalar \glspl{ADC} $\tilde{P}$ and quantization resolution $\tilde{M}$, given by $\tilde{M} = \left\lfloor 2^{\frac{M_{\text{bit}}}{\tilde{P}L}} \right\rfloor$. This trade-off can also be quantified with the analog combining ratio $r \triangleq \tilde{P}/N_r$ with a fixed $R$ \cite{shlezinger2019asymptotic}, which will be investigated with numerical results in Section \ref{sec:result}.
\end{remark}

In \gls{ISAC} systems, the transmitted signal $\mymat{\Theta}$ is a random but known signal rather than a fixed pilot, motivating us to consider two strategies of designing $\mymat{A}$ and $\mymat{B}$. The first is to design $\mymat{A}, \mymat{B}$ depending on each different $\mymat{\Theta}$ and the expectation in \eqref{eqn:mse_g} is \gls{wrt} the channel $\myvec{g}$ and receiver noise $\myvec{w}$. This strategy is similar to \cite{shlezinger2019asymptotic} and called \gls{DD} in this paper. As shown in the sequel, however, solving the optimal $\mymat{A}$ involves a convex optimization, which increase the hardware and computational complexity, making it more difficult for real-time implementation. Therefore, we consider the second strategy where $\mymat{A}$ is fixed and independent of $\mymat{\Theta}$ and the expectation in \eqref{eqn:mse_g} is \gls{wrt} to not only $\myvec{g}$ and $\myvec{w}$ but also $\mymat{\Theta}$. Such strategy is called \gls{DI}. To facilitate a uniform notation for \gls{DD} and \gls{DI} in the sequel, we introduce the following operator:
\vspace{-0.2cm}
\begin{equation} \label{eqn:operator}
E_{D,i}(x) =
\begin{cases}
x, &\quad i = 0,\\
\myexpecttheta{x}, &\quad i = 1, \\
\end{cases}
\vspace{-0.2cm}
\end{equation}
where $x$ is a scalar, vector or matrix function of $\mymat{\Theta}$. The case of $i=0$ corresponds to the \gls{DD} strategy since it treats $\mymat{\Theta}$ as a deterministic variable, while the case of $i=1$ corresponds to the \gls{DI} strategy since it treats $\mymat{\Theta}$ as a random variable. Moreover, for the conciseness of notation, we will drop $i$ in this operator and use $E_{D}$ instead when a uniform notation for both strategies is possible.

For both strategies, the matrices $\mymat{A}$ and $\mymat{B}$ that minimize the average \gls{MSE} are stated in the following theorem:

\begin{theorem} \label{theorem:opt}
	When $K_d\geq 2$ and $\eta < \sqrt{3/(2K_d)}\tilde{M}$, the optimal analog combining matrix $\mymat{A}$ for the \gls{DD} and \gls{DI} strategies can be given by $\mymat{A} = \mymat{U\Lambda V}^H \mymat{R}_A^{-\frac{1}{2}}$, where $\mymat{U}\in \mathbb{C}^{\tilde{P} \times \tilde{P}}$ is the unitary \gls{DFT} matrix, i.e.,
	\vspace{-0.2cm}
	\begin{equation} \label{eqn:dft}
	 (\mymat{U})_{p,q} = \frac{1}{\sqrt{\tilde{P}}} e^{-j2\pi(p-1)(q-1)/\tilde{P}}.
	 \vspace{-0.2cm}
	\end{equation}
The matrix $\mymat{V}^H \in \mathbb{C}^{N_r \times N_r}$ is the eigenmatrix of $\mymat{R}_A$, i.e., $\mymat{\Lambda}_A = \mymat{V}^H \mymat{R}_A \mymat{V} \in \mathbb{C}^{N_r \times N_r}$ is a diagonal matrix, whose diagonal entries are the eigenvalues of $\mymat{R}_A$, denoted by $\{\lambda_{A,i}\}_{i=1}^{N_r}$ and assumed to be sorted in the descending order. The matrix $\mymat{\Lambda} \in \mathbb{C}^{\tilde{P}\times N_r}$ is a diagonal matrix with non-negative diagonal entries $\{\bar\sigma_{i}\}_{i=1}^{\tilde{P}}$. In the \gls{DD} strategy, $\{\bar\sigma_{i}\}_{i=1}^{\tilde{P}}$ is the solution to the following convex optimization:
\vspace{-0.2cm}
	\begin{equation} \label{eqn:opt_eig_DD}
	\begin{aligned}
	\left\{\bar{\sigma}_i\right\}_{i=1}^{\tilde{P}}= & \underset{\left\{\sigma_i\right\}_{i=1}^{\tilde{P}}}{\arg \max } \sum_{i=1}^{\tilde{P}} \sum_{n_t=1}^{N_t} 
	\frac{d_{B, n_t} \lambda_{A,i} \lambda'_{n_t}\sigma^2_i}{(\lambda'_{n_t}+\sigma_w^2) \sigma_i^2 + \beta}, \\
	& \text { subject to }  \sum_{i=1}^{\tilde{P}} \sigma_i^2 = 1, 
	\end{aligned}
	\vspace{-0.2cm}
	\end{equation}
	while in the \gls{DI} strategy, $\{\bar\sigma_{i}\}_{i=1}^{\tilde{P}}$ is the solution to the following convex optimization:
	\begin{equation} \label{eqn:opt_eig_DI}
	\begin{aligned}
	\left\{\bar{\sigma}_i\right\}_{i=1}^{\tilde{P}}= & \underset{\left\{\sigma_i\right\}_{i=1}^{\tilde{P}}}{\arg \max } \mathbb{E}_{\mymat{\Theta}}\left\{\sum_{i=1}^{\tilde{P}} \sum_{n_t=1}^{N_t} \frac{\lambda_{A,i} \lambda'_{n_t}\sigma^2_i}{(\lambda'_{n_t}+\sigma_w^2) \sigma_i^2 + \beta}\right\}, \\
	& \text { subject to }  \sum_{i=1}^{\tilde{P}} \sigma_i^2 = 1, 
	\end{aligned}
	\end{equation}
	where $\beta = \frac{2(K_d+1)\kappa \sigma_{\max}^2}{3\tilde{M}^2 \tilde{P}}$, $\kappa = \eta^2 \left(1-\frac{2 K_d\eta^2}{3\tilde{M}^2}\right)^{-1}$ and
	\begin{equation} \label{eqn:sigma_max}
	\sigma_{\max}^2 = \mytr\left( \mymat{R}_\theta^* \mymat{R}_B \right) + \sigma_w^2.
	\end{equation}
%	\begin{equation}
%	\sigma_{\max}^2 \triangleq 
%	\begin{cases}
%	\underset{l = 1,\dots, L}{\max} \left(\mymat{\Theta}^T \mymat{R}_B \mymat{\Theta}^*\right)_{l,l} + \sigma_w^2, \text{ \gls{DD}},
%	\\
%	\mytr\left( \mymat{R}_\theta^* \mymat{R}_B \right) + \sigma_w^2, \text{ \gls{DI}}.
%	\end{cases}
%	\end{equation}
	The eigenvalues of random matrix $\mymat{R}_B^{\frac{1}{2}} \mymat{\Theta}^* \mymat{\Theta}^T (\mymat{R}_B^{\frac{1}{2}})^H$ construct $\{\lambda'_{n_t}\}_{n_t = 1}^{N_t}$.  In \eqref{eqn:opt_eig_DD},  $\{d_{B,i}\}_{i=1}^{N_t}$ consists of the non-negative diagonal entries of $\mymat{U}'^H \mymat{R}_B \mymat{U}'$, where $\mymat{U}'$ is the eigenmatrix of $\mymat{R}_B^{\frac{1}{2}} \mymat{\Theta}^* \mymat{\Theta}^T (\mymat{R}_B^{\frac{1}{2}})^H$.

	With $\mymat{A}$ provided for both \gls{DD} and \gls{DI}, the corresponding optimal $\mymat{B}$ is given by
	\begin{equation} \label{eqn:opt_B}
	\begin{aligned}
		&\mymat{B} = \left(\mymat{R}_B \mymat{\Theta}^* \otimes \mymat{R}_A {\mymat{A}}^H\right) \\
		&\times \!\left(\!\left(\!\left(\!\mymat{\Theta}^T \mymat{R}_B \mymat{\Theta}^* \!+\! \sigma_w^2\mymat{I}_{L}\right)\! \otimes \!{\mymat{A}} \mymat{R}_A {\mymat{A}}^H\right) \!+ \!
		\frac{2(K_d+1) \gamma^2}{3 \tilde{M}^2} \mymat{I}_{L}\!\right)^{-1},
	\end{aligned}
	\end{equation}
	where $\gamma = \sqrt{\kappa/\tilde{P}} \cdot \sigma_{\max}$.
\end{theorem}
\begin{proof}
	See Appendix \ref{app:proof_opt_quan}.
	\vspace{-0.2cm}
\end{proof}
\begin{remark}
	\label{remark:opt_quantizer}
	In Theorem \ref{theorem:opt}, we impose $K_d\geq 2$ so that the  quantization noise is uncorrelated to the input \cite{gray1993dithered}, thereby facilitating our analysis. However, Theorem \ref{theorem:opt} can also serve as a method to approximate the optimal quantizer in the case of $K_d = 0 \text{ or } 1$, and we will evaluate the sensing performance with numerical results by applying Theorem \ref{theorem:opt} to the case of $K_d = 0$, corresponding to no dither, in Section \ref{sec:result}. The condition $\eta < \sqrt{3/(2K_d)}\tilde{M}$ is imposed so that there exists a positive support $\gamma$ such that \eqref{eqn:support_multiple} holds. In the case of $K_d=0$, this inequality holds for any positive $\eta$.
\end{remark}

According to Theorem \ref{theorem:opt}, the optimal $\mymat{\Lambda}$ can be obtained by solving the convex optimization $\eqref{eqn:opt_eig_DD}$ in the \gls{DD} strategy. However, in the \gls{DI} strategy, although the optimization \eqref{eqn:opt_eig_DI} is also convex, it is difficult to directly solve \eqref{eqn:opt_eig_DI} since the objective function involves the expectation \gls{wrt} the random signal $\mymat{\Theta}$ and it is intractable to express the expectation explicitly. Therefore, we refer to the \gls{SAA} method \cite{kim2015guide} to solve this stochastic optimization, which will be described below.

First, we take $N_s$ samples of $\mymat{\Theta}$, denoted by $\mymat{\Theta}_1, \dots, \mymat{\Theta}_{N_s}$. Then, we get the eigenvalues of each $\mymat{R}_B^{\frac{1}{2}} \mymat{\Theta}_{n_s}^* \mymat{\Theta}_{n_s}^T (\mymat{R}_B^{\frac{1}{2}})^H,$ $n_s = 1,\dots, N_s$, denoted by $\{\lambda'_{n_s, n_t}\}_{n_t = 1}^{N_t}$.  The objective function of \eqref{eqn:opt_eig_DI} can be approximated as
\vspace{-0.2cm}
\begin{equation}
\begin{aligned}
&\mathbb{E}_{\mymat{\Theta}}\left\{
\sum_{i=1}^{\tilde{P}} \sum_{n_t=1}^{N_t} \frac{\lambda_{A,i} \lambda'_{n_t}\sigma^2_i}{(\lambda'_{n_t}+\sigma_w^2) \sigma_i^2 + \beta}
\right\}
\\ \approx  &
\frac{1}{N_s} 
\sum_{n_s=1}^{N_s} \sum_{i=1}^{\tilde{P}} \sum_{n_t=1}^{N_t} \frac{\lambda_{A,i} \lambda'_{n_s,n_t}\sigma^2_i}{(\lambda'_{n_s,n_t}+\sigma_w^2) \sigma_i^2 + \beta}
\\
= &\frac{1}{N_s} 
\sum_{n=1}^{N_s N_t} \sum_{i=1}^{\tilde{P}}
 \frac{\lambda_{A,i} \tilde{\lambda}'_{n}\sigma^2_i}{(\tilde{\lambda}'_{n}+\sigma_w^2) \sigma_i^2 + \beta},
\end{aligned}
\vspace{-0.2cm}
\end{equation}
where $\{\tilde{\lambda}'_{n}\}_{n = 1}^{N_s N_t}$ contains all $\lambda'_{n_s, n_t}$ by setting $\tilde{\lambda}'_{(n_s-1) N_t + n_t} = \lambda'_{n_s, n_t}$. Therefore, the optimization of $\mymat{\Lambda}$ in the \gls{DI} strategy with \gls{SAA} method is given by the following convex optimization:
\vspace{-0.2cm}
\begin{equation} \label{eqn:opt_eig_DI_SAA}
\begin{aligned}
\left\{\bar{\sigma}_i\right\}_{i=1}^{\tilde{P}}= & \underset{\left\{\sigma_i\right\}_{i=1}^{\tilde{P}}}{\arg \max } \frac{1}{N_s} 
\sum_{n=1}^{N_s N_t} \sum_{i=1}^{\tilde{P}}
\frac{\lambda_{A,i} \tilde{\lambda}'_{n}\sigma^2_i}{(\tilde{\lambda}'_{n}+\sigma_w^2) \sigma_i^2 + \beta}, \\
& \text { subject to }  \sum_{i=1}^{\tilde{P}} \sigma_i^2 = 1. 
\end{aligned}
\vspace{-0.2cm}
\end{equation}

\section{Numerical Results}
\label{sec:result}
In this section, we evaluate the performance of \gls{ISAC} sensing systems that utilize quantizers proposed in Section \ref{sec:quantizer_structure} with numerical results. First, we introduce the parameter configurations. Then, we investigate the relationship between sensing performance and analog combining ratio $r$ to illustrate the trade-off between quantization output size $\tilde{P}$ and quantization resolution $\tilde{M}$, as addressed in Remark \ref{remark:tradeoff}. Finally, we compare the performance of the proposed \gls{DD} and \gls{DI} quantization strategies with other benchmarks with different quantization rate $R$.

We consider an \gls{ISAC} \gls{BS} performing \gls{MIMO} sensing. The number of transmitting and receiving antennas are $N_t = 6$ and $N_r = 20$. The number of snapshots for each \gls{TIR} sensing is $L = 40$. The spatial correlation matrices at the receiving and transmission array, i.e., $\mymat{R}_A$ and $\mymat{R}_B$, follow the Jakes model \cite{jakes1994microwave} by setting $\left(\mymat{R}_A\right)_{n_1,n_2} = J_0(\pi\abs{n_1-n_2})$ and $\left(\mymat{R}_B\right)_{n_1,n_2} = J_0(0.8\pi\abs{n_1-n_2})$, where $J_0$ is the zero-order Bessel function of the first type. The signal $\mymat{\Theta}$ is generated with $\mymat{\Theta} = \mymat{W}_{\text{pre}} \mymat{\Theta}_0$, where $\mymat{W}_{\text{pre}} \in \mathbb{C}^{N_t\times N_t}$ is a fixed precoding matrix and each entry of $\mymat{\Theta}_0$ follows $\mathcal{CN}(0,1)$ in an \gls{iid} manner. Thus, the correlation of each snapshot is $\mymat{R}_{\theta} =  \mymat{W}_{\text{pre}} \mymat{W}_{\text{pre}}^H$. The noise variance is $\sigma_w^2 = 10^{-3}$.

As discussed in Remark \ref{remark:tradeoff}, the analog combining ratio $r$ influences the sensing performance, which is illustrated by the results in Fig.~\ref{fig:mse_g_combining_ratio}. First, the case of no quantization is simulated as a benchmark. When designing the quantizer, we fix $\eta = 2$. The number of scalar \glspl{ADC} $\tilde{P}$ varies from $1$ to $N_r$, corresponding to $r$ from $0.05$ to $1$. The quantizers are designed with the \gls{DD} and \gls{DI} strategies, colored as blue and red in Fig.~\ref{fig:mse_g_combining_ratio}, respectively. In the \gls{DI} strategy, we set $N_s=10^4$ when solving \eqref{eqn:opt_eig_DI_SAA}. The quantization rate $R$ is set as 2 and 4. Additionally, as stated in Remark \ref{remark:opt_quantizer}, Theorem \ref{theorem:opt} can also serve as a method to approximate the optimal quantizer in the case of $K_d = 0 \text{ or } 1$. In Fig.~\ref{fig:mse_g_combining_ratio}, therefore, we examine the cases of $K_d=2$ and $K_d = 0$ (corresponding to no dither).  Each point on the curves is obtained by performing $N_{\text{sim}} = 1000$ Monte Carlo experiments and calculating the average \gls{MSE} of channel estimation $\sigma_g^2$ as in \eqref{eqn:mse_g}. The average \gls{MSE}  with no quantization is $1.60\times 10^{-3}$. The following discussions and observations are made on Fig.~\ref{fig:mse_g_combining_ratio}: First, in the case of $K_d=2$ and $R=2$, corresponding to circle markers, the combining ratio $r$ is no greater than 0.6. The reason is that when $r>0.6$, the value of $\kappa = \eta^2 \left(1-\frac{2 K_d\eta^2}{3\tilde{M}^2}\right)^{-1}$ is negative, meaning that there does not exist a support $\gamma$ of scalar \glspl{ADC} for \eqref{eqn:support_multiple} to hold, as stated in Remark \ref{remark:opt_quantizer}. Second, for the same quantizer design strategy and $R$, the average \gls{MSE} without dither is less than that with dither. The reason is that the addition of dither signals increases the quantization noise level of each \gls{ADC}. Third, the average \gls{MSE} of the \gls{DI} strategy is very close to that of the \gls{DD} strategy, meaning that the \gls{DI} strategy can greatly reduce the hardware and computational complexity at the price of a slight degradation of sensing performance. Fourth, the average \gls{MSE} reaches its minimum with $r=1$ except the case of $R=2$ with dither. That is, no dimension compression is preferred in terms of sensing performance when using the spatial analog combining, which is similar to the results in \cite{shlezinger2019asymptotic} and will be further examined in the next numerical results.

\begin{figure}[htb]
	\begin{minipage}[b]{1.0\linewidth}
		\centering
		\centerline{\includegraphics[width=8cm]{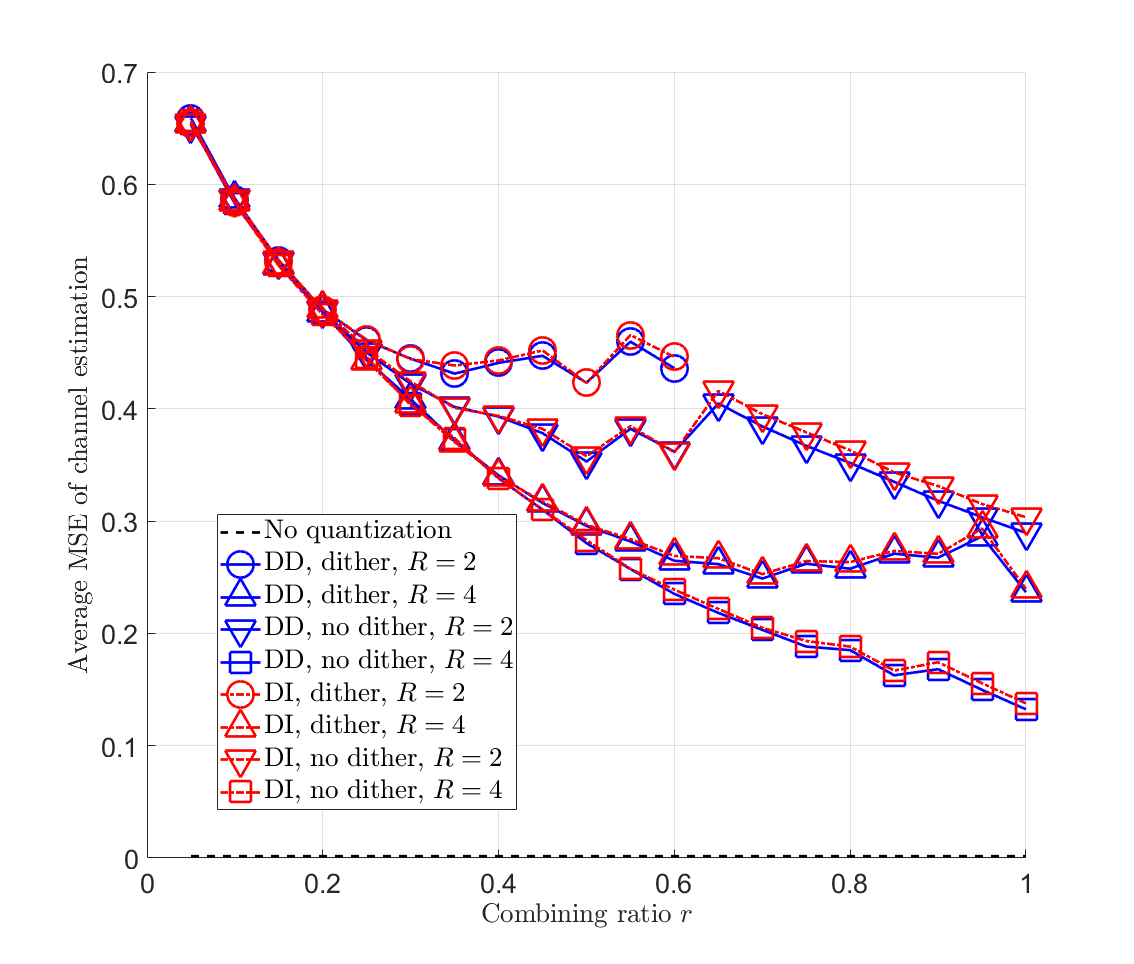}}
	\end{minipage}
	\caption{Average \gls{MSE} $\sigma_g^2$ as a function of combining ration $r$ with different quantizer design strategies, dither schemes and quantization ratios.}% The conventional and optimized design are marked as the cross and circle.}
	\label{fig:mse_g_combining_ratio}
\end{figure}

Next, we evaluate the sensing performance with varying quantization rate $R$ and different quantizer design methods. The simulated rates $R$ are the integers from 2 to 16. We only consider the case of $K_d = 0$, i.e., no dither, in both \gls{DD} and \gls{DI} strategies to achieve better sensing performance. By conducting simulations similar to Fig.~\ref{fig:mse_g_combining_ratio}, we find that the optimal combining ratio $r$ is $1$ for any $R$ in the no dither quantizers, indicating that no dimension compression is preferred with spatial analog combining again. Additionally, we also evaluate the quantization structure where no analog spatial combining is performed, which is equivalent to $\mymat{A} = \mymat{I}_{N_r}$ and also referred to as task-ignorant digital only quantization \cite{shlezinger2019asymptotic}. The following discussions and observations are made on Fig.~\ref{fig:mse_g_quantization_rate}: First, as in Fig.~\ref{fig:mse_g_combining_ratio}, the average \gls{MSE} of the \gls{DI} strategy is very close to that of \gls{DD}, especially with high quantization rates. Second, the proposed task-based \gls{DD} and \gls{DI} quantizations outperform the task-ignorant digital only quantization in terms of sensing performance for most of the quantization rates. Specifically,  when $R$ is from $5$ to $13$, the average \gls{MSE} can be reduced by $1.1\sim 1.5$ dB with the \gls{DI} strategy compared with digital-only.

\begin{figure}[htb]
	\begin{minipage}[b]{1.0\linewidth}
		\centering
		\centerline{\includegraphics[width=8cm]{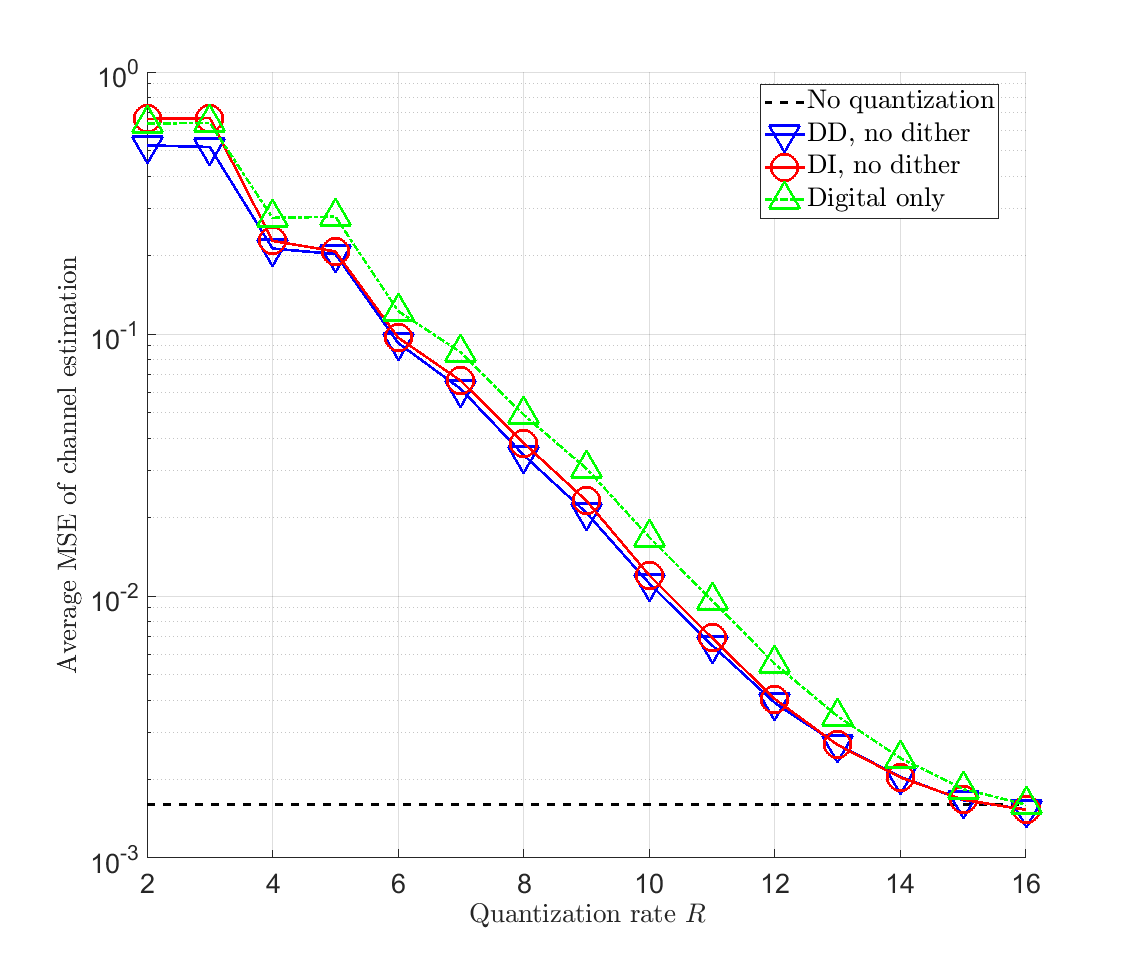}}
	\end{minipage}
	\caption{Average \gls{MSE} $\sigma_g^2$ as a function of combining ration $r$ with different quantizer design strategies, dither schemes and quantization ratios.}% The conventional and optimized design are marked as the cross and circle.}
	\label{fig:mse_g_quantization_rate}
\end{figure}

\section{Conclusion and Discussion}
\label{sec:discussion}

In this paper, we address the challenge of quantizer design in MIMO ISAC systems with random signals. We propose two strategies for quantizer optimization: data-dependent (DD) and data-independent (DI). Both strategies aims at minimizing the MSE of TIR estimation, adapting the quantizer design to the random signaling of ISAC systems. We theoretically derive the optimal quantizers for both strategies and propose an \gls{SAA}-based algorithm to solve the optimization problem in \gls{DI} strategy.
Our results demonstrates that the DI strategy, despite its lower computational complexity compared to the DD strategy, achieved near-optimal sensing performance. This finding is significant as it suggests that practical ISAC systems can employ efficient quantizer designs without substantial performance degradation. The study also reveals that the proposed quantizers outperform digital-only quantization in terms of sensing performance. 

%This work has several implications for the design of future ISAC systems. Firstly, the proposed quantization strategies can be integrated into existing systems to enhance their sensing capabilities. Secondly, the insights on the trade-offs involved in quantizer design provide guidelines for system engineers to optimize hardware configurations based on available resources.
%
%Several avenues for future work can be considered. Extending the quantizer design to more complex ISAC scenarios, such as dynamic environments or systems with higher-order MIMO configurations, would be valuable. Additionally, investigating the impact of different types of random signaling on quantizer performance could further refine the design principles. Lastly, exploring the integration of machine learning techniques into quantizer design could potentially lead to adaptive and more efficient ISAC systems.
\bibliographystyle{IEEEtran}
\bibliography{IEEEabrv,2D}

\begin{thebibliography}{10}
\providecommand{\url}[1]{#1}
\csname url@samestyle\endcsname
\providecommand{\newblock}{\relax}
\providecommand{\bibinfo}[2]{#2}
\providecommand{\BIBentrySTDinterwordspacing}{\spaceskip=0pt\relax}
\providecommand{\BIBentryALTinterwordstretchfactor}{4}
\providecommand{\BIBentryALTinterwordspacing}{\spaceskip=\fontdimen2\font plus
\BIBentryALTinterwordstretchfactor\fontdimen3\font minus
  \fontdimen4\font\relax}
\providecommand{\BIBforeignlanguage}[2]{{%
\expandafter\ifx\csname l@#1\endcsname\relax
\typeout{** WARNING: IEEEtran.bst: No hyphenation pattern has been}%
\typeout{** loaded for the language `#1'. Using the pattern for}%
\typeout{** the default language instead.}%
\else
\language=\csname l@#1\endcsname
\fi
#2}}
\providecommand{\BIBdecl}{\relax}
\BIBdecl

\bibitem{liu2022integrated}
F.~Liu, Y.~Cui, C.~Masouros, J.~Xu, T.~X. Han, Y.~C. Eldar, and S.~Buzzi,
  ``Integrated sensing and communications: Toward dual-functional wireless
  networks for 6g and beyond,'' \emph{IEEE J. Sel. Areas Commun.}, vol.~40,
  no.~6, pp. 1728--1767, 2022.

\bibitem{gao2022integrated}
Z.~Gao, Z.~Wan, D.~Zheng, S.~Tan, C.~Masouros, D.~W.~K. Ng, and S.~Chen,
  ``Integrated sensing and communication with {mmWave} massive {MIMO}: A
  compressed sampling perspective,'' \emph{IEEE Trans. Wireless Commun.},
  vol.~22, no.~3, pp. 1745--1762, 2022.

\bibitem{lu2023random}
S.~Lu, F.~Liu, F.~Dong, Y.~Xiong, J.~Xu, Y.-F. Liu, and S.~Jin, ``Random {ISAC}
  signals deserve dedicated precoding,'' \emph{arXiv preprint
  arXiv:2311.01822}, 2023.

\bibitem{zhang2021overview}
J.~A. Zhang, F.~Liu, C.~Masouros, R.~W. Heath, Z.~Feng, L.~Zheng, and
  A.~Petropulu, ``An overview of signal processing techniques for joint
  communication and radar sensing,'' \emph{IEEE J. Sel. Topics Signal
  Process.}, vol.~15, no.~6, pp. 1295--1315, 2021.

\bibitem{xie2023novel}
Q.~Xie, C.~Liu, Z.~Mo, and W.~Li, ``A novel pulse-agile waveform design based
  on random {FM} waveforms for range sidelobe suppression and range ambiguity
  mitigation,'' \emph{IEEE Trans. Geosci. Remote Sensing}, vol.~61, pp. 1--12,
  2023.

\bibitem{gray1998quantization}
R.~M. Gray and D.~L. Neuhoff, ``Quantization,'' \emph{IEEE Trans. Inf. Theory},
  vol.~44, no.~6, pp. 2325--2383, 1998.

\bibitem{shlezinger2019asymptotic}
N.~Shlezinger, Y.~C. Eldar, and M.~R. Rodrigues, ``Asymptotic task-based
  quantization with application to massive {MIMO},'' \emph{IEEE Trans. Signal
  Process.}, vol.~67, no.~15, pp. 3995--4012, 2019.

\bibitem{ma2021bit}
D.~Ma, N.~Shlezinger, T.~Huang, Y.~Liu, and Y.~C. Eldar, ``Bit constrained
  communication receivers in joint radar communications systems,'' in
  \emph{2021 IEEE Int. Conf. Acoust. Speech Signal Process.}\hskip 1em plus
  0.5em minus 0.4em\relax IEEE, 2021, pp. 8243--8247.

\bibitem{shlezinger2019hardware}
N.~Shlezinger, Y.~C. Eldar, and M.~R. Rodrigues, ``Hardware-limited task-based
  quantization,'' \emph{IEEE Trans. Signal Process.}, vol.~67, no.~20, pp.
  5223--5238, 2019.

\bibitem{weichselberger2006stochastic}
W.~Weichselberger, M.~Herdin, H.~Ozcelik, and E.~Bonek, ``A stochastic {MIMO}
  channel model with joint correlation of both link ends,'' \emph{IEEE Trans.
  Wireless Commun.}, vol.~5, no.~1, pp. 90--100, 2006.

\bibitem{kermoal2002stochastic}
J.-P. Kermoal, L.~Schumacher, K.~I. Pedersen, P.~E. Mogensen, and
  F.~Frederiksen, ``A stochastic {MIMO} radio channel model with experimental
  validation,'' \emph{IEEE J. Sel. Areas Commun.}, vol.~20, no.~6, pp.
  1211--1226, 2002.

\bibitem{gray1993dithered}
R.~Gray and T.~Stockham, ``Dithered quantizers,'' \emph{IEEE Trans. Inf.
  Theory}, vol.~39, no.~3, pp. 805--812, 1993.

\bibitem{widrow1996statistical}
B.~Widrow, I.~Kollar, and M.-C. Liu, ``Statistical theory of quantization,''
  \emph{IEEE Trans. Instrum. Meas.}, vol.~45, no.~2, pp. 353--361, 1996.

\bibitem{kim2015guide}
S.~Kim, R.~Pasupathy, and S.~G. Henderson, ``A guide to sample average
  approximation,'' \emph{Handbook of Simulation Optimization}, pp. 207--243,
  2015.

\bibitem{jakes1994microwave}
W.~C. Jakes and D.~C. Cox, \emph{Microwave mobile communications}.\hskip 1em
  plus 0.5em minus 0.4em\relax Wiley-IEEE press, 1994.

\bibitem{palomar2007mimo}
D.~P. Palomar and Y.~Jiang, \emph{{MIMO} Transceiver Design via Majorization
  Theory}.\hskip 1em plus 0.5em minus 0.4em\relax Delft, The Netherlands: Now
  Publishers, 2007.

\bibitem{couillet2011random}
R.~Couillet and M.~Debbah, \emph{Random Matrix Methods for Wireless
  Communications}.\hskip 1em plus 0.5em minus 0.4em\relax Cambridge University
  Press, 2011.

\end{thebibliography}
% Appendices
\clearpage
\newpage
\appendices
\section{Proof of Theorem \ref{theorem:opt}}
\label{app:proof_opt_quan}
To prove Theorem \ref{theorem:opt}, we first characterize the \gls{MMSE} of \gls{TIR} estimation. Second, we derive  the optimal unitary rotation $\mymat{U}$ for a given $\mymat{A}$ as in \cite[Appendix C]{shlezinger2019hardware}. Third, we find the solution to the optimal $\mymat{V}$ and $\mymat{\Lambda}$ with a combination of \gls{SVD} and convex optimization.

We first fix $\mymat{A}$ and find the corresponding optimal $\mymat{B}$. When $K_d \geq 2$ and the input plus dither signals is within the quantizer's support $\gamma$, the output of dithered quantizer is the sum of the input and an uncorrelated white quantization noise with variance $\frac{2 (K_d+1) \gamma^2}{3 \tilde{M}^2} $ \cite{gray1993dithered}. Therefore, for a given $\mymat{A}$, the digital processing matrix yielding linear \gls{MMSE} estimation is given by \eqref{eqn:opt_B} \cite[Proposition 4]{shlezinger2019asymptotic} and the corresponding average \gls{MSE} is 
\begin{equation} \label{eqn:mse_g_theta}
\begin{aligned}
& \sigma^2_{g|\Theta}\left(\mymat{A}\right) 
=\frac{1}{N_t N_r} \operatorname{Tr}\left(\mymat{R}_B \otimes \mymat{R}_A \right) \\
&-\frac{1}{N_t N_r} \operatorname{Tr}\left(\left(\mymat{\Theta}^T \mymat{R}_B^2 \mymat{\Theta}^* \otimes \mymat{A} \mymat{R}_A^2 \mymat{A}^H\right)\right. \\
& \left.\times\!\left(\!\left(\!\left(\!\mymat{\Theta}^T \mymat{R}_B \mymat{\Theta}^* \!+\! \sigma_w^2\mymat{I}_{L}\!\right) \!\otimes\! {\mymat{A}} \mymat{R}_A {\mymat{A}}^H \!\right)\!+\!
\frac{2 (K_d\!+\!1) \gamma^2}{3 \tilde{M}^2} 
\mymat{I}_{L \tilde{P}}\!\right)^{-1}\!\right).
\end{aligned}
\end{equation}
Note that $\sigma^2_{g|\Theta}$ in \eqref{eqn:mse_g_theta} is different from  $\sigma^2_{g}$ in \eqref{eqn:mse_g} in that $\sigma^2_{g|\Theta}$ depends on the signal $\mymat{\Theta}$. Thus, we have $\sigma^2_{g} = E_D\left(\sigma^2_{g|\Theta}\right)$ for both strategies. % $\sigma^2_{g} = \sigma^2_{g|\Theta}$ in the \gls{DD} strategy while $\sigma^2_{g} = \myexpecttheta{\sigma^2_{g|\Theta}}$ in the \gls{DI} strategy.

Next, we obtain the support $\gamma$ of the quantizer which is decided by \eqref{eqn:support_multiple}. To facilitate this, we first derive the correlation of $\myvec{y}$ as follows:
\begin{equation} \label{eqn:cor_y}
\begin{aligned}
\mymat{\Sigma}_y & \triangleq\myexpect{\myvec{y}\myvec{y}^H}  = \left( \myexpect{\mymat{\Theta}^T \mymat{R}_B \mymat{\Theta}^*} + \sigma_w^2 \mymat{I}_L \right) \otimes \mymat{R}_A \\
&  {=} \mymat{\Sigma}_0 \otimes \mymat{R}_A, 
\end{aligned}
\end{equation}
where
\begin{equation}
\begin{aligned}
\mymat{\Sigma}_0 \triangleq \myexpect{\mymat{\Theta}^T \mymat{R}_B \mymat{\Theta}^*} + \sigma_w^2 \mymat{I}_L 
\overset{(a)}{=} \sigma_{\max}^2 \mymat{I}_L,
%\begin{cases}
%{\mymat{\Theta}^T \mymat{R}_B \mymat{\Theta}^*} + \sigma_w^2 \mymat{I}_L, \text{ \gls{DD}},\\
% \left( \mytr\left( \mymat{R}_\theta^* \mymat{R}_B \right) + \sigma_w^2 \right) \mymat{I}_L , \text{ \gls{DI}}.
%\end{cases}
\end{aligned}
\end{equation}
%Further, in the \gls{DI} strategy, we have
where $\sigma_{\max}^2$ is provided in \eqref{eqn:sigma_max} and $(a)$ holds since
\begin{equation} 
\begin{aligned}
&\myexpect{\left(\mymat{\Theta}^T \mymat{R}_B \mymat{\Theta}^*\right)_{l_1,l_2}} 
= \myexpect{\myvec{\theta}_{l_1}^T \mymat{R}_B \myvec{\theta}_{l_2}^*}\\
= &\myexpect{\mytr\left(\myvec{\theta}_{l_2}^*\myvec{\theta}_{l_1}^T \mymat{R}_B \right)}
= \mytr\left(\mymat{R}_{\theta}^* \mymat{R}_B \right) \delta(l_1-l_2),
\end{aligned}
\end{equation}
%Then, $\mymat{\Sigma}_y$ in the \gls{DI} strategy can be transformed into
where $\delta(l)$ is the Kronecker delta function. Then, the correlation of $\myvec{u}$ is given by
\begin{equation} \label{eqn:cor_u}
\mymat{\Sigma}_u \triangleq \myexpect{\myvec{u}\myvec{u}^H} 
= \mymat{\Sigma}_0
\otimes \mymat{A}\mymat{R}_A\mymat{A}^H.
\end{equation}
According to \cite[Equation (D.3)]{shlezinger2019asymptotic}, in order for \eqref{eqn:support_multiple} to hold, we have
\begin{equation} \label{eqn:support_result}
\begin{aligned}
\gamma^2 &= \kappa \max_{i = 1,\dots, L\tilde{P}} \left(\mymat{\Sigma}_u\right)_{i,i} = \kappa \sigma^2_{\max} \underset{i = 1,\dots, \tilde{P}}{\max} \left(\mymat{A}\mymat{R}_A\mymat{A}^H\right)_{i,i},
\end{aligned}
\end{equation}
where $\kappa = \eta^2 \left(1-\frac{2 K_d\eta^2}{3\tilde{M}^2}\right)^{-1}$. 
Defining $\mymatb{A} = \mymat{A}\mymat{R}_A^{\frac 12}$ and substituting \eqref{eqn:support_result} into \eqref{eqn:mse_g_theta} yields
%\begin{equation} \label{eqn:mse_g_theta2}
%\begin{aligned}
%& \sigma^2_g\left(\mymatb{A}\right) = \myexpectd{\sigma^2_{g|\Theta}\left(\mymatb{A}\right)} =\frac{1}{N_t N_r} \operatorname{Tr}\left(\mymat{R}_B \otimes \mymat{R}_A \right)-\frac{1}{N_t N_r} \operatorname{E_D}\left(\operatorname{Tr}\left(\left(\mymat{\Theta}^T \mymat{R}_B^2 \mymat{\Theta}^* \otimes \mymatb{A} \mymat{R}_A \mymatb{A}^H\right)\right.\right. \\
%& \left.\left.\times\left(\left(\left(\mymat{\Theta}^T \mymat{R}_B \mymat{\Theta}^* + \sigma_w^2\mymat{I}_{L}\right) \otimes {\mymatb{A}} {\mymatb{A}}^H\right)+\frac{4 \kappa \sigma_{\max} \underset{i = 1,\dots, \tilde{P}}{\max} \left(\mymatb{A}\mymatb{A}^H\right)_{i,i}}{3 \tilde{M}^2} \mymat{I}_{L \tilde{P}}\right)^{-1}\right)\right).
%\end{aligned}
%\end{equation}
\begin{equation} 
\label{eqn:mse_g_theta2}
\begin{aligned}
& \sigma^2_g\left(\mymatb{A}\right) = \myexpectd{\sigma^2_{g|\Theta}\left(\mymatb{A}\right)} =\frac{1}{N_t N_r} \operatorname{Tr}\left(\mymat{R}_B \otimes \mymat{R}_A \right) \\
& -\frac{1}{N_t N_r} \operatorname{E_D}\Biggl(\operatorname{Tr}\Biggl(\Bigl(\mymat{\Theta}^T \mymat{R}_B^2 \mymat{\Theta}^* \otimes \mymatb{A} \mymat{R}_A \mymatb{A}^H\Bigr) \\
& \times\Biggl(\Bigl(\mymat{\Theta}^T \mymat{R}_B \mymat{\Theta}^* + \sigma_w^2\mymat{I}_{L}\Bigr) \otimes {\mymatb{A}} {\mymatb{A}}^H \\
& +\frac{2(K_d+1) \kappa \sigma^2_{\max} \underset{i = 1,\dots, \tilde{P}}{\max} \left(\mymatb{A}\mymatb{A}^H\right)_{i,i}}{3 \tilde{M}^2} \mymat{I}_{L \tilde{P}}\Biggr)^{-1}\Biggr)\Biggr).
\end{aligned}
\end{equation}
According to \cite[Lemma D.1]{shlezinger2019asymptotic}, for any matrix $\mymatb{A}$, there exists a unitary matrix $\mymat{U}\in \mathbb{C}^{\tilde{P}\times\tilde{P}}$ such that $\mymat{U}\mymatb{A}\mymatb{A}^H\mymatb{U}^H$ is weakly majorized by all possible rotations of $\mymatb{A}\mymatb{A}^H$, i.e., all diagonal entries of $\mymat{U}\mymatb{A}\mymatb{A}^H\mymatb{U}^H$ are equal to $\frac{1}{\tilde{P}} \mytr\left(\mymatb{A}\mymatb{A}^H\right)$. As a result, $\max_{i = 1,\dots, \tilde{P}}\left(\mymat{U}\mymatb{A}\mymatb{A}^H\mymatb{U}^H\right)_{i,i}$ in \eqref{eqn:mse_g_theta2} is also equal to  $\frac{1}{\tilde{P}} \mytr\left(\mymatb{A}\mymatb{A}^H\right)$. Such a matrix  $\mymat{U}$ yields the minimum $\sigma^2_{g}$ among all the rotations of  $\mymatb{A}\mymatb{A}^H$, corresponding to $\sigma^2_{g}$ given by
\begin{equation} \label{eqn:mse_g_theta_min}
\begin{aligned}
\sigma^2_{g}\left(\mymat{U}\mymatb{A}\right) 
 =\frac{1}{N_t N_r} \operatorname{Tr}\left(\mymat{R}_B \otimes \mymat{R}_A \right)-\frac{1}{N_t N_r} f_A,
\end{aligned}
\end{equation}
where
\begin{equation} \label{eqn:obj_func}
\begin{aligned}
& f_A =  \operatorname{E_D}\left(\operatorname{Tr}\left(\left(\mymat{\Theta}^T \mymat{R}_B^2 \mymat{\Theta}^* \otimes \mymatb{A} \mymat{R}_A \mymatb{A}^H\right)\right. \right.\\
& \left.\left.\cdot\!\left(\left(\left(\mymat{\Theta}^T \mymat{R}_B \mymat{\Theta}^* \!+\! \sigma_w^2\mymat{I}_{L}\right) \!\otimes\! \mymatb{A}\mymatb{A}^H\right)\!+\!\beta \mytr\left(\mymatb{A}\mymatb{A}^H\right) \mymat{I}_{L \tilde{P}}\right)^{-1}\right)\right),
\end{aligned}
\end{equation}
where $\beta = \frac{2(K_d+1)\kappa \sigma_{\max}^2}{3\tilde{M}^2 \tilde{P}}$. 
We can now optimize $\mymatb{A}$ such that \eqref{eqn:mse_g_theta_min} reaches its minimum, and after the optimal $\mymatb{A}$ is obtained, we will show that the unitary \gls{DFT} matrix given by \eqref{eqn:dft} is the optimal $\mymat{U}$ among all rotations. We perform \gls{SVD} on $\mymatb{A}$, i.e., $\mymatb{A} = \mymat{U}_A \mymat{\Lambda} \mymat{V}^H$, where $\mymat{U}_A \in \mathbb{C}^{\tilde{P}\times \tilde{P}}$ and $\mymat{V} \in \mathbb{C}^{N_r \times N_r}$ are unitary matrices, and $\mymat{\Lambda} \in \mathbb{C}^{\tilde{P} \times N_r}$ is a diagonal matrix with non-negative diagonal entries $\{\sigma_{i}\}_{i=1}^{\tilde{P}}$. Substituting this \gls{SVD} into \eqref{eqn:mse_g_theta_min}, we can transform the minimization of \eqref{eqn:mse_g_theta_min} into 
\begin{equation}
\label{eqn:opt_svd}
\begin{aligned}
\underset{\mymat{\Lambda}, \mymat{V}}{\max } f_A(\mymat{\Lambda}, \mymat{V}), %\operatorname{Tr}\left(\left(\Theta^T D_{l, l}^4 \Theta^* \otimes \bar{A} C_l \overline{\boldsymbol{A}}^H\right)\right. \\
%\left.\qquad\left(\left(\Sigma_{\boldsymbol{y}_l} \otimes \overline{\boldsymbol{A}} \overline{\boldsymbol{A}}^H\right)+\boldsymbol{I}_{\bar{P}_L}\right)^{-1}\right), %\\
%& \text { subject to } \frac{4 \kappa \cdot \sigma_l^2}{3 \tilde{M}^2 \cdot \tilde{P}} \operatorname{Tr}\left(\overline{\boldsymbol{A}} \overline{\boldsymbol{A}}^H\right)=1 .
\end{aligned}
\end{equation}
where $f_A$ can be transformed into
\begin{equation} \label{eqn:obj_func_svd}
\begin{aligned}
& f_A =  \operatorname{E_D}\left(\operatorname{Tr}\left(\left(\mymat{\Theta}^T \mymat{R}_B^2 \mymat{\Theta}^* \otimes \mymat{\Lambda}\mymat{V}^H \mymat{R}_A \mymat{V}\mymat{\Lambda}^H\right)\right. \right.\\
& \left.\left.\cdot\!\left(\left(\left(\mymat{\Theta}^T \mymat{R}_B \mymat{\Theta}^* \!+\! \sigma_w^2\mymat{I}_{L}\right) \!\otimes\! \mymat{\Lambda}\mymat{\Lambda}^H\right)\!+\!\beta \mytr\left(\mymat{\Lambda}\mymat{\Lambda}^H\right) \mymat{I}_{L \tilde{P}}\right)^{-1}\right)\right).
\end{aligned}
\end{equation}
It is indicated by \eqref{eqn:obj_func_svd} that $f_A$ is invariant of the left singular matrix $\mymat{U}_A$. Thus, $\mymat{U}_A$ is dropped in the optimization \eqref{eqn:opt_svd}.

To make the maximization of $f_A$ further tractable, we perform \gls{SVD} on $\mymat{R}_B^{\frac{1}{2}} \mymat{\Theta}^*$, i.e., $\mymat{R}_B^{\frac{1}{2}} \mymat{\Theta}^* = \mymat{U}' \mymat{\Lambda}' \mymat{V}'^H$, where $\mymat{U}' \in \mathbb{C}^{{N_t}\times N_t}$ and $\mymat{V}' \in \mathbb{C}^{L \times L}$ are unitary matrices, and $\mymat{\Lambda}' \in \mathbb{C}^{N_t \times L}$ is a diagonal matrix with non-negative diagonal entries $\{\sigma'_{i}\}_{i=1}^{N_t}$. Then, $f_A$ can be transformed into
\begin{equation} \label{eqn:obj_func_svd2}
\begin{aligned}
& f_A =  \operatorname{E_D}\left(\operatorname{Tr}\left(\left(\mymat{U}'^H \mymat{R}_B \mymat{U}' \otimes \mymat{V}^H \mymat{R}_A \mymat{V}\right)\mymat{M}\right)\right),
\end{aligned}
\end{equation}
where
\begin{equation} \label{eqn:diag_M}
\begin{aligned}
 \mymat{M} \!=\! &  
\left(\mymat{\Lambda}' \!\otimes\! \mymat{\Lambda}^H\right)\!
\left(\!\left(\!\left(\!\mymat{\Lambda}'^H\! \mymat{\Lambda}' \!+\! \sigma_w^2\mymat{I}_{L}\right) \!\otimes\! \mymat{\Lambda}\mymat{\Lambda}^H\right)\!+\!\beta \mytr\!\left(\!\mymat{\Lambda}\!\mymat{\Lambda}^H\!\right)\! \mymat{I}_{L \tilde{P}}\!\right)\!^{-1}\\
&\times\left(\mymat{\Lambda}'^H \otimes \mymat{\Lambda}\right).
\end{aligned}
\end{equation}
Note that $\mymat{M}$ is a diagonal matrix, which enables us to transform $f_A$ into the following form:
\begin{equation} \label{eqn:obj_func_sum}
\begin{aligned}
f_A = \operatorname{E_D}\left(\sum_{n_t=1}^{N_t}\sum_{i = 1}^{\tilde{P}}
\frac{d_{B, n_t} d_{A,i} \lambda'_{n_t} {\sigma}_i^2}
{(\lambda'_{n_t}+\sigma_w^2){\sigma}_i^2 + \beta\sum_{q=1}^{\tilde{P}} {\sigma}_q^2}
\right),
\end{aligned}
\end{equation}
where $\{d_{B,i}\}_{i=1}^{N_t}$ and $\{d_{A,i}\}_{i=1}^{N_r}$ are the non-negative diagonal entries of $\mymat{U}'^H \mymat{R}_B \mymat{U}' $ and $\mymat{V}^H \mymat{R}_A \mymat{V}$, respectively. We define $\lambda'_{n_t} \triangleq \sigma'^2_{n_t}$ and  $\{\lambda'_{n_t}\}_{n_t = 1}^{N_t}$ can be interpreted as the eigenvalues of $\mymat{R}_B^{\frac{1}{2}} \mymat{\Theta}^* \mymat{\Theta}^T (\mymat{R}_B^{\frac{1}{2}})^H$. Then, it follows
from \cite[Theorem II.1]{palomar2007mimo} that \eqref{eqn:obj_func_sum} is maximized by setting $\mymat{V}$ to
be the eigenmatrix of $\mymat{R_A}$ that yields $\mymat{V}^H \mymat{R}_A \mymat{V}$ as a diagonal matrix with non-negative diagonal entries $\{\lambda_{A,i}\}_{i=1}^{N_r}$ in the descending order. Then, we also have $d_{A,i}= \lambda_{A,i}, i=1,\dots, N_r$.

Now that the optimal $\mymat{V}$ in \eqref{eqn:opt_svd} is decided, the remaining step is to optimize $\mymat{\Lambda}$, or equivalently, the diagonal entries $\{\sigma_{i}\}_{i=1}^{\tilde{P}}$ based on the objective function given by \eqref{eqn:obj_func_sum}. Note that in \eqref{eqn:obj_func_sum},  $\{d_{B,i}\}_{i=1}^{N_t}$ depend on the left singular matrix of $\mymat{R}_B^{\frac{1}{2}} \mymat{\Theta}^*$ and $\{\lambda'_{n_t}\}_{n_t = 1}^{N_t}$ are eigenvalues of $\mymat{R}_B^{\frac{1}{2}} \mymat{\Theta}^* \mymat{\Theta}^T (\mymat{R}_B^{\frac{1}{2}})^H$. Since the signal $\mymat{\Theta}$ is treated with different manners in \gls{DD} and \gls{DI} strategies, the transformation of optimizing $\mymat{\Lambda}$ is also split according to the two strategies. 

In the \gls{DD} strategy,  $\{d_{B,i}\}_{i=1}^{N_t}$ and $\{\lambda'_{n_t}\}_{n_t = 1}^{N_t}$ are known deterministic variables. Additionally, note that the value of $f_A$ remains unchanged after multiplying $\{\sigma_{i}\}_{i=1}^{\tilde{P}}$ by the same positive scalar $\alpha$. Thus, we can assume that $\sum_{i=1}^{\tilde{P}} {\sigma}_i^2 = 1$ in its optimization. Then, by replacing $\{d_{A,i}\}_{i=1}^{N_r}$ with the optimal solution $\{\lambda_{A,i}\}_{i=1}^{N_r}$ in \eqref{eqn:obj_func_sum}, the optimization of $\{\sigma_{i}\}_{i=1}^{\tilde{P}}$ in the \gls{DD} strategy is given by \eqref{eqn:opt_eig_DD}.

In the \gls{DI} strategy, $\{d_{B,i}\}_{i=1}^{N_t}$ and $\{\lambda'_{n_t}\}_{n_t = 1}^{N_t}$ are random variables. By replacing $\{d_{A,i}\}_{i=1}^{N_r}$ with $\{\lambda_{A,i}\}_{i=1}^{N_r}$ and assuming that $\sum_{i=1}^{\tilde{P}} {\sigma}_i^2 = 1$, $f_A$ can be further transformed as follows:
\begin{equation} \label{eqn:obj_func_DI}
\begin{aligned}
f_A &= \myexpecttheta{\sum_{n_t=1}^{N_t}\sum_{i = 1}^{\tilde{P}}
\frac{d_{B, n_t} \lambda_{A,i} \lambda'_{n_t} {\sigma}_p^2}
{(\lambda'_{n_t}+\sigma_w^2){\sigma}_i^2 + \beta}}
,\\
& \overset{(a)}{=}\sum_{n_t=1}^{N_t}\sum_{i = 1}^{\tilde{P}}
\myexpecttheta{d_{B, n_t}}
\myexpecttheta{
	\frac{ \lambda_{A,i} \lambda'_{n_t} {\sigma}_p^2}
	{(\lambda'_{n_t}+\sigma_w^2){\sigma}_i^2 + \beta}} \\
& \overset{(b)}{=}\sum_{n_t=1}^{N_t}\sum_{i = 1}^{\tilde{P}}
\frac{\mytrop{\mymat{R}_B}}{N_t}
\myexpecttheta{
	\frac{ \lambda_{A,i} \lambda'_{n_t} {\sigma}_i^2}
	{(\lambda'_{n_t}+\sigma_w^2){\sigma}_i^2 + \beta}}.
\end{aligned}
\end{equation}
Here, $(a)$ holds since $\mymat{U}'$ and $\lambda'_{n_t} \triangleq \sigma'^2_{n_t}$ are independent as they can been seen as the eigenmatrix and eigenvalues of the Wishart matrix $\mymat{R}_B^{\frac{1}{2}} \mymat{\Theta}^* \mymat{\Theta}^T (\mymat{R}_B^{\frac{1}{2}})^H$, respectively \cite[Theorem 2.2]{couillet2011random}. Equality $(b)$ holds since
\begin{equation}
\begin{aligned}
\myexpecttheta{d_{B, n_t}} 
&= 
\myexpecttheta{\left(\mymat{U}'^H \mymat{R}_B \mymat{U}' \right)_{n_t,n_t}}\\
&=
\myexpecttheta{\myvec{u}_{n_t}'^H \mymat{R}_B \myvec{u}_{n_t}'}
=
\myexpecttheta{\mytrop{\myvec{u}_{n_t}'^H \mymat{R}_B \myvec{u}_{n_t}'}}\\
&=
\mytrop{ \mymat{R}_B \myexpecttheta{\myvec{u}_{n_t}' \myvec{u}_{n_t}'^H}}
\overset{(c)}{=} \frac{\mytrop{\mymat{R}_B}}{N_t},
\end{aligned}
\end{equation}
where $\myvec{u}_{n_t}'$ denotes the $n_t-$th column of $\mymat{U}'$, i.e., the eigenvector of $\mymat{R}_B^{\frac{1}{2}} \mymat{\Theta}^* \mymat{\Theta}^T (\mymat{R}_B^{\frac{1}{2}})^H$, and $(c)$ holds since the eigenvectors of Wishart matrix $\mymat{R}_B^{\frac{1}{2}} \mymat{\Theta}^* \mymat{\Theta}^T (\mymat{R}_B^{\frac{1}{2}})^H$ are
uniformly distributed on the unit sphere \cite[Theorem 2.2]{couillet2011random}, resulting in $\myexpecttheta{\myvec{u}_{n_t}' \myvec{u}_{n_t}'^H} =  \mymat{I}_{N_t}/N_t$.
Then, by leaving out the constant term $\frac{\mytrop{\mymat{R}_B}}{N_t}$ in the last line of \eqref{eqn:obj_func_DI}, the optimization of $\{\sigma_{i}\}_{i=1}^{\tilde{P}}$ in the \gls{DD} strategy is given by \eqref{eqn:opt_eig_DI}.

Both \eqref{eqn:opt_eig_DD} and \eqref{eqn:opt_eig_DI} are convex optimizations since $\frac{ \lambda_{A,i} \lambda'_{n_t} {\sigma}_i^2}
{(\lambda'_{n_t}+\sigma_w^2){\sigma}_i^2 + \beta}$ is a concave function \gls{wrt} ${\sigma}_i^2$.

Now that $\mymat{V}$ and $\mymat{\Lambda}$ are optimized, we prove that the unitary \gls{DFT} matrix given by \eqref{eqn:dft} is the optimal $\mymat{U}$ among all rotations. Recalling that \eqref{eqn:obj_func_svd} indicates that $f_A$ is invariant of the left singular matrix of $\mymatb{A}$, we can set $\mymatb{A}$ with the optimized $\mymat{V}, \mymat{\Lambda}$ and find the unitary matrix $\mymat{U}$ such that all diagonal entries of $\mymat{U}\mymatb{A}\mymatb{A}^H\mymatb{U}^H$ are equal. In this setting, $\mymatb{A}\mymatb{A}^H = \mymatb{\Lambda}\mymatb{\Lambda}^H$ is a diagonal matrix. Thus, $\mymat{U}\mymatb{A}\mymatb{A}^H\mymatb{U}^H$ has equal diagonal entries when $\mymat{U}$ is the unitary \gls{DFT} matrix \cite[Lemma 2.10]{palomar2007mimo}.

\end{document}